\documentclass[a4paper,runningheads]{llncs}
\usepackage[utf8]{inputenc}
\usepackage{thm-restate}
\usepackage{graphicx,amssymb,amsmath,url, ntheorem}
\usepackage[hidelinks]{hyperref}
\usepackage{amsfonts,cleveref}
\usepackage{mathtools}
\usepackage{subcaption}
\usepackage{mathrsfs}
\usepackage{hyperref}
\usepackage{xspace}
\usepackage{cite}
\usepackage{tikz}
\usepackage{tikz-qtree}
\usepackage{todonotes}
\usepackage{complexity}
\usepackage{wrapfig}
\usepackage{etex}
\usepackage{apptools}
\usepackage{xpatch}

\newcommand\blfootnote[1]{%
	\begingroup
	\renewcommand\thefootnote{}\footnote{#1}%
	\addtocounter{footnote}{-1}%
	\endgroup
}

\makeatletter
\xpatchcmd{\thmt@restatable}%
{\csname #2\@xa\endcsname\ifx\@nx#1\@nx\else[{#1}]\fi}%
{\IfAppendix{\csname #2\@xa\endcsname}{\csname #2\@xa\endcsname\ifx\@nx#1\@nx\else[{#1}]\fi}}%
{}{} %
\makeatother

\let\doendproof\endproof
\renewcommand\endproof{\null \hfill $\qed$\doendproof}

\author{Henry Förster\inst{1}\orcidID{0000-0002-1441-4189} \and Robert Ganian\inst{2}\orcidID{0000-0002-7762-8045} \and Fabian Klute\inst{2}\orcidID{0000-0002-7791-3604} \and Martin N\"ollenburg\inst{2}\orcidID{0000-0003-0454-3937}}
\title{On Strict (Outer-)Confluent Graphs}

\institute{
	University of Tübingen, Tübingen, Germany
	\email{foersth@informatik.uni-tuebingen.de}
	\and 
	Algorithms and Complexity Group, TU Wien, Vienna, Austria
	\email{\{rganian,fklute,noellenburg\}@ac.tuwien.ac.at}
}

\begin{document}
	\nocite{eppstein2005delta}
	\nocite{yu1995efficient}
	\nocite{trotter2001combinatorics}
	\nocite{halldorsson2011alternation}
	
	\maketitle
	
	\begin{abstract}
		A strict confluent (SC) graph drawing is a drawing of a graph with vertices as points in the plane, where vertex adjacencies are represented not by individual curves but rather by unique smooth paths through a planar system of junctions and arcs.
		If all vertices of the graph lie in the outer face of the drawing, the drawing is called a strict outerconfluent (SOC) drawing.
		SC and SOC graphs were first considered by Eppstein et al.\ in Graph Drawing 2013.
		Here, we  establish several new relationships between the class of SC graphs and other graph classes, in particular string graphs and unit-interval graphs. 
		Further, we extend earlier results about special bipartite graph classes to the notion of strict outerconfluency, show that SOC graphs have cop number two, and establish that tree-like ($\Delta$-)SOC graphs have bounded cliquewidth.		
		\blfootnote{A poster containing some of the results of this paper was presented at GD 2017. Robert Ganian acknowledges support by the Austrian Science Fund (FWF, project P31336) and is also affiliated with FI MUNI, Brno, Czech Republic.}
	\end{abstract}
	
	\section{Introduction}\label{sec:introduction}
	Confluent drawings of graphs are geometric graph representations in the Euclidean plane, in which vertices are mapped to points, but edges are not drawn as individually distinguishable geometric objects. 
	Instead, an edge between two vertices $u$ and $v$ is represented by 
	a smooth path between the points of $u$ and $v$ through a crossing-free system of arcs and junctions. 
	Since multiple edge representations may share some arcs and junctions of the drawing, this allows dense and non-planar graphs to be drawn in a plane way (e.g., see Fig.~\ref{fig:junctions} for a confluent drawing of $K_5$). 
	Hence confluent drawings can be seen as theoretical counterpart of heuristic edge bundling techniques, which are frequently used in network visualizations to reduce visual clutter in layouts of dense graphs~\cite{brhmd-tueb-17,h-hebvarhd-06}.
	
	More formally, a \emph{confluent drawing} $D$ of a graph $G=(V,E)$ consists of a set of points representing the vertices of $ G $, a set of junction points, and a set of smooth arcs, such that each arc starts and ends at either a vertex point or a junction, no two arcs intersect (except at common endpoints), and all arcs meeting in a junction share the same tangent line in the junction point.
	There is an edge $(u,v) \in E$ if and only if there is a smooth path from $u$ to $v$ in $D$ not passing through any other vertex.
	
	\begin{figure}[tb]
		\centering
		\includegraphics{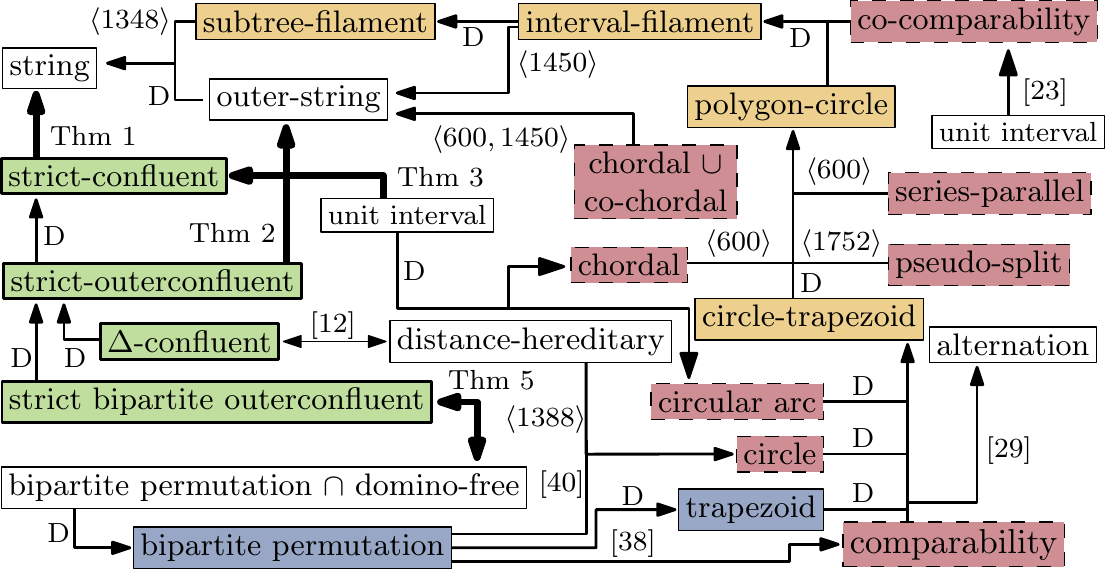}
		\caption{Inclusions among graph classes related to SOC graphs. Arrows point from sub- to superclass, where edge label `D' marks an inclusion by definition. Fat arrows are inclusions shown in this paper and are labelled with the corresponding theorem. Green boxes are confluent graph classes. Red, dashed boxes are classes that are incomparable to SOC graphs. Orange boxes are classes that are potential superclasses of SOC graphs. Blue boxes are potential subclasses of the SOC graphs. The numbers in $ \langle \cdot \rangle $ indicate references of \url{graphclasses.org}.
		}\label{fig:hierarchy}
	\end{figure}	
	
	Confluent drawings were introduced by Dickerson et al.~\cite{degm-cdvndp-05}, who identified classes of graphs that admit or do not admit confluent drawings. Subsequently, the notions of strong and tree confluency have been introduced~\cite{hui2007train}, as well as $\Delta$-confluency~\cite{eppstein2005delta}. Confluent drawings have further been used for drawings of layered graphs~\cite{egm-cld-07} and Hasse diagrams~\cite{eppstein2013hasse}. 
	Eppstein et al.~\cite{eppstein2016strict} defined the class of strict confluent (SC) drawings, which require that every edge of the graph must be represented by a unique smooth path and that there are no self-loops. They showed that for general graphs it is \NP-complete to decide whether an SC drawing exists.
	An SC drawing is called \emph{strict outerconfluent} %
	(SOC) if all vertices lie on the boundary of a (topological) disk that contains the SC drawing.
	For graphs with a given cyclic vertex order, Eppstein et al.~\cite{eppstein2016strict} presented a constructive efficient algorithm for testing the existence of an SOC drawing.
	Without a given vertex order, neither the recognition complexity nor a characterization of such graphs is known.

	We approach the characterization problem by comparing the SOC graph class with a hierarchy of classes of intersection graphs. In general a \emph{geometric intersection graph} $G=(V,E)$ is a graph with a bijection between the vertices $V$ and a set of geometric objects such that two objects intersect if and only if the corresponding vertices are adjacent. Common examples include interval graphs, string graphs~\cite{ehrlich1976intersection}
	and circle graphs~\cite{gabor1989recognizing}.
	Since confluent drawings make heavy use of intersecting curves to represent edges in a planar way, it seems natural to ask what kind of geometric intersection models can represent a confluent graph.

	\medskip
	
	\noindent{\bf Contributions.}
	After introducing basic definitions and properties in Section~\ref{sec:preliminaries}, we show in Section~\ref{sec:string} that SC and SOC graphs are, respectively, string and outerstring graphs\cite{k-sgincngi-91}.
	Section~\ref{sec:interval} shows that every unit interval graph~\cite{roberts1969indifference,wegner1967eigenschaften} can be drawn strict confluent.
	In Section~\ref{sec:bipartite}, we consider the so-called strict bipartite-outerconfluent drawings: by following up on an earlier result of Hui et al.~\cite{hui2007train}, we show that graphs which admit such a drawing are precisely the domino-free bipartite permutation graphs. 
	Inspired by earlier work of Gaven\v ciak et al.~\cite{DBLP:conf/isaac/GavenciakJKK13}, we examine in Section~\ref{sec:cops} the cop number of SOC graphs and show that it is at most two. %
	In Appendix~\ref{sec:nonincl} we show that many natural subclasses of outer-string graphs are incomparable to SOC graphs (see red, dashed boxes in Fig.~\ref{fig:hierarchy}).
	More specifically, we show that 
	circle \cite{gabor1989recognizing},
	circular-arc \cite{hsu1985maximum},
	series-parallel \cite{takamizawa1982linear},
	chordal \cite{gavril1972algorithms},
	co-chordal \cite{benzaken1990more}, and co-comparability \cite{GOLUMBIC198337}
	graphs are all incomparable to SOC graphs. This list may help future research by excluding a series of natural candidates for sub- and super-classes of SOC graphs. Finally, in Section~\ref{sec:cliquewidth}, we show that the cliquewidth of so-called tree-like $ \Delta $-SOC graphs is bounded by a constant, generalizing a previous result of Eppstein et al.~\cite{eppstein2005delta}.

	\section{Preliminaries}\label{sec:preliminaries}
	
	A \emph{confluent diagram} $ D = (N,J,\Gamma) $ in the plane $\mathbb R^2$ consists of a set $ N $ of points called \emph{nodes}, a set $ J $ of points called \emph{junctions} and a set $ \Gamma $ of simple smooth curves called \emph{arcs} whose endpoints are in $ J \cup N $. Further, two arcs may only intersect at common endpoints. If they intersect in a junction they must share the same tangent line, see Fig~\ref{fig:junctions}.
	
	\begin{wrapfigure}[13]{r}{.4\textwidth}
		\centering
		\vspace{-9mm}
		\includegraphics[scale=1]{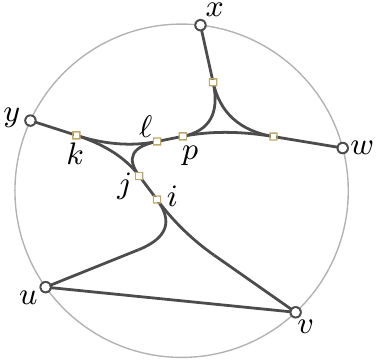}
		\caption{A strict outerconfluent diagram representing $K_5$. Nodes are disks, junctions are squares.}
		\label{fig:junctions}
	\end{wrapfigure}
	
	Let $ D = (N,J,\Gamma) $ be a confluent diagram and let $ u,v \in N $ be two nodes. A \emph{$uv$-path} $ p =  (\gamma_0,\dots,\gamma_k) $ %
	in $ D $ is a sequence of arcs 
	$ \gamma_0 = (u,j_1), \gamma_1=(j_1,j_2), \dots, 
	\gamma_k=(j_{k},v) \in \Gamma $
	such that $ j_1,\dots j_{k} $ are junctions and $ p $ is a smooth curve. %
	In Fig.~\ref{fig:junctions} the unique  $uy$-path passes through junctions $i,j,k$.
	If there is at most one $uv$-path for each pair of nodes $u,v$ in $N$ and if there are no self-loops, i.e., no $uu$-path for any $u \in N$, we say that $D$ is a \emph{strict} confluent diagram.
	The uniqueness of $uv$-paths and the absence of self-loops imply that every $uv$-path is actually a path in the graph-theoretic sense, where no vertex is visited twice.
	We further define $ P(D) $ as the set of all smooth paths between all pairs of nodes in $ N $. Let $ p \in P(D) $ be a path and $ j \in J $ a junction in $ D $, then we write $ j \in p $, if $ p $ passes through $ j $.

	As observed by Eppstein et al.~\cite{eppstein2016strict}, we may assume that every junction is a \emph{binary} junction, where exactly three arcs meet such that the three enclosed angles are $180^\circ, 180^\circ, 0^\circ$.
	In other words two arcs from the same direction merge into the third arc, or, conversely, one arc splits into two arcs.
	A (strict) confluent diagram with higher-degree junctions can easily be transformed into an equivalent (strict) one with only binary junctions.

	Let $j \in J$ be a binary junction with the three incident arcs $\gamma_1, \gamma_2, \gamma_3$. 
	Let the angle enclosed by $\gamma_1$ and $\gamma_2$ be $0^\circ$ and the angle enclosed by $\gamma_3$ and $\gamma_1$ (or $\gamma_2$) be $180^\circ$. 
	Then we say that $j$ is a \emph{merge-junction} for $\gamma_1$ and $\gamma_2$ and a \emph{split-junction} for $\gamma_3$. 
	We also say that $\gamma_1$ and $\gamma_2$ \emph{merge} at $j$ and that  $\gamma_3$ \emph{splits} at $j$.
	Given two nodes $u,v \in N$ and a junction $j \in J$ we say that $j$ is a merge-junction for $u$ and $v$ if there is a third node $w \in N$, a $uw$-path $p$ and a $vw$-path $q$ such that $j \in p$ and $j \in q$, the respective incoming arcs $\gamma_p = (j_p,j)$ and $\gamma_q = (j_q,j)$ are distinct and the suffix paths of $p$ and $q$ from $j$ to $w$ are equal. 
	Conversely, we say that a junction $j \in J$ is a split-junction for a  node $u \in N$ if there are two nodes $v,w \in N$, a $uv$-path $p$, and a $uw$-path $q$ such that $j \in p$ and $j \in q$, the prefix paths of $p$ and $q$ from $u$ to $j$ are equal and the respective subsequent arcs $\gamma_p = (j,j_p)$ and $\gamma_q = (j,j_q)$ are distinct.
	In Fig.~\ref{fig:junctions}, junction $i$ is a merge-junction for $u$ and $v$, while it is a split junction for each of $w,x,y$. Two junctions $ i,j \in J $ are called a \emph{merge-split pair} if $ i $ and $ j $ are connected by an arc $\gamma$ and both $i$ and $j$ are split-junctions for $\gamma$; in Fig.~\ref{fig:junctions}, junctions $i$ and $j$ form a merge-split pair, as well as junctions $\ell$ and $p$.

	We call an arc $ \gamma \in \Gamma $ \emph{essential} if we cannot delete $ \gamma $ without changing adjacencies in the represented graph. We call a confluent diagram $ D $ \emph{reduced}, if every arc is essential. 
	Notice that this is a different notion than strictness, since it is possible that in a confluent diagram we find two essential arcs between a pair of nodes.
	Without loss of generality we can assume that the nodes of an outerconfluent diagram are placed on a circle with all arcs and junctions inside the circle. We can infer a \emph{cyclic order} $ \pi $ from an outerconfluent diagram $ D $ by walking clockwise around the boundary of the unbounded face and adding the nodes to $\pi$ in the order they are visited. 
	
	From a confluent diagram $ D = (N,J,\Gamma) $ we derive a simple, undirected graph $G_D = (V_D, E_D) $ with $V_D = N$ and $E_D = \{ (u,v) \mid \exists\, uv\text{-path } p \in P(D) \}$. We say $D$ is a confluent drawing of a graph $G$ if $G$ is isomorphic to $G_D$ and that $G$ is a (strict) (outer-)confluent graph if it admits a (strict) (outer-)confluent drawing.
	
	\section{Strict (Outer-)Confluent $\subset$ (Outer-)String}\label{sec:string}
	The class of \emph{string graphs}~\cite{k-sgincngi-91} contains all graphs $ G = (V,E) $ which can be represented as the intersection graphs of open curves in the plane. We show that they form a superclass of SC graphs and
	that every SOC graph is an outer-string graph~\cite{k-sgincngi-91}. \emph{Outer-string} graphs are string graphs that can be represented so that strings lie inside a disk and intersect the boundary of the disk in one endpoint. 
	Note that strings are allowed to self-intersect and cross each more than once.
	
	Let $ D = (N,J,\Gamma) $ be a strict confluent diagram. For every node $ u \in N $ we construct the \emph{junction tree} $ T_u $ of $ u $, with root $u$ and a leaf for each neighbor $v$ of $ u $ in $ G_D $. The interior vertices of $ T_u $ are the junctions which lie on the (unique) $uv$-paths. The strictness of $D$ implies that $T_u$ is a tree. Observe that every internal node of $T_u$ has at most two children. Further, every merge-junction for $ u $ is a vertex with one child in $ T_u $, and every split-junction for $u$ has two children. For every junction $ j $ in $T_u$ we can define the sub-tree $ T_{u,j} $ of $T_u$ with root $ j $.

	\begin{restatable}{lemma}{obsindepsubtrees}
		\label{obs:indep_subtrees}
		Let $ D = (N,J,\Gamma) $ be a strict confluent diagram, let $ u,v \in N $ be two nodes and let $i, j$ be two distinct merge-junctions for $u,v$. Then $i$ is neither an ancestor nor a descendant of $j$ in $T_u$ (and, by symmetry, in $T_v$).
	\end{restatable}
	To create a string representation of an SC graph we trace the paths of a strict confluent diagram $ D = (N,J,\Gamma) $, starting from each node $ u \in N $ and combine them into a string representation.  Figure~\ref{fig:traces} shows an example. We traverse the junction tree for each $ u \in N $ on the left-hand side of each arc (seen from its root $u$) and create  a string $ t(u) $, the \emph{trace} of $ u $, with respect to $ T_u $ as follows.
	
	Start from $ u $ and traverse $T_u$ in left-first DFS order.
	Upon reaching a leaf $ \ell $ make a clockwise U-turn and backtrack to the previous split-junction of $T_u$.
	When returning to a split-junction we have two cases. (a) coming from the left subtree: cross the arc from the left subtree at the junction and descend into the right subtree. (b) coming from the right subtree: cross the arc to the left subtree again and backtrack upward in the tree along the existing trace to the previous split-junction of $T_u$.
	
	\begin{wrapfigure}[12]{r}{.55\textwidth}
		\vspace{-1.05cm}
		\centering
		\includegraphics[scale=1]{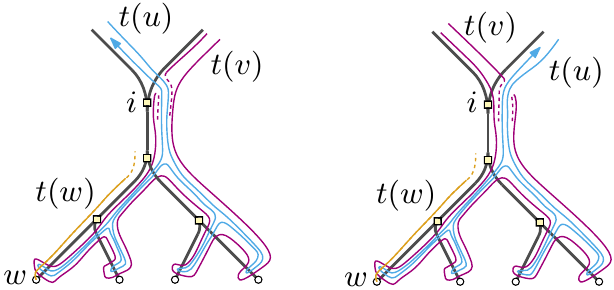} 
		\caption{Two possible configurations for inserting a new trace $t(u)$ that meets an existing trace $t(v)$ at a merge junction $i$; $t(v)$ is cut 
			and re-routed.}
		\label{fig:traces}
	\end{wrapfigure}
	
	Finally, at a merge-junction $i$ with at least one trace from the other arc merging into $i$ already drawn:
	Let $ v \in N $  such that $ u $ and $ v $ merge at $ i $ and $t(v)$ is already tracing the subtree $T_{u,i} = T_{v,i}$. 
	In this case we temporarily cut open the part of trace $t(v)$ closest to $t(u)$, route $t(u)$ through the gap and let it follow  $t(v)$ along $T_{u,i}$ until it returns to junction $i$, where $t(u)$ passes through the gap again. 
	Since $T_{u,i} = T_{v,i}$ this is possible without $t(u)$ intersecting $t(v)$.
	Now it remains to reconnect the two open ends of $t(v)$, but this can again be done without any new intersections by winding $t(v)$ along the ``outside'' of $t(u)$. See Fig.~\ref{fig:traces} for an illustration. 
	If there are multiple traces with this property, they can all be treated as a single ``bundled'' trace within $T_{u,i}$.

	\begin{restatable}{theorem}{thmsc}
		\label{thm:sc}
		Every SC graph is a string graph. 
	\end{restatable}
	\begin{proof}
		Given an SC graph $ G = (V,E) $ with a strict confluent drawing $ D = (N,J,\Gamma) $ we construct the traces as described above for every node $ u \in N $. 	In the following let $ u,v $ be two nodes of $ D $. We distinguish three cases.
		
		\textbf{Case 1} ($ uv $-path in $ P(D) $): %
		We draw $t(u)$ and $ t(v) $ as described above. Since there is a $uv$-path in $ P(D) $ we have to guarantee that $ t(u) $ and $ t(v) $ intersect at least once. We introduce crossings at the leaves corresponding to $ u $ and $ v $ in $ T_u $ and $ T_v $ when $t(u)$ and $t(v)$ make a U-turn; see how the trace $t(u)$ 
		intersects $t(w)$ near the leaf $w$ in Fig.~\ref{fig:traces}.

		\textbf{Case 2} (No $ uv $-path in $ P(D) $ and $ u,v $ share no merge-junction): In this case $ T_u $ and $ T_v $ are disjoint trees. Traces can meet only at shared junctions and around leaves, but since $t(u)$ and $t(v)$ trace disjoint trees
		intersections are impossible. 
		
		\textbf{Case 3} (No $ uv $-path in $ P(D) $ and $ u,v $ share a merge-junction): First assume $ u $ and $ v $ share a single merge-junction $ i \in J $ and assume $t(v)$ is already drawn when creating trace $t(u)$. We have to be careful that $ t(v) $ and $ t(u) $ do not intersect. If we route the traces at the merge-junction $i$ as depicted in Fig.~\ref{fig:traces}, they visit the shared subtree $T_{u,i}=T_{v,i}$ without intersecting each other.

		Now assume $ u $ and $ v $ share $k>1$ merge-junctions $ j_1,\dots,j_k \in J $ and $ u $ and $ v $ merge at each $ j_i $. Consequently we find $ k $ shared subtrees $ T^1,\dots, T^k $. By Lemma~\ref{obs:indep_subtrees}, however, we know that the intersection of these subtrees is empty. Hence we can treat every merge-junction and its subtree independently as in the case of a single merge-junction.
		
		These are all the cases how two junction trees can interact. Hence the traces $t(u)$ and $t(v)$ for nodes $ u,v \in N $ intersect if and only if there is a $uv$-path in $ P(D) $ and, equivalently, the edge $ (u,v) \in E_D $. Further, every trace is a continuous curve, so this set of traces yields a string representation of $ G $. 
	\end{proof}

	A construction following the same principle can in fact be used to show: 
	\begin{restatable}{theorem}{thmsoc}
		\label{thm:soc}
		Every SOC graph is an outer-string graph. 
	\end{restatable}

	\section{Unit Interval Graphs and SC}\label{sec:interval}
	In this section we consider so-called unit interval graphs. Let $ G = (V,E) $ be a graph, then $ G $ is a unit interval graph if there exists a \emph{unit-interval} layout $ \Gamma_{UI} $ of $ G $, i.e. a representation of $ G $ where each vertex $ v \in V $ is represented as an interval of unit length and edges are given by the intersections of the intervals.
	
	\begin{restatable}{theorem}{unitIntervalInSC}
		\label{thm:unitIntervalInSC}
		Every unit-interval graph is an SC graph.
	\end{restatable}

	\begin{proof}[Sketch]
		Our proof technique is constructive and describes how to compute a strict confluent diagram $D$ for a given graph $G$ based on its unit-interval layout $\Gamma_{UI}$. Based on the ordering of intervals in $\Gamma_{UI}$, we first greedily compute a set of cliques which are subgraphs of $G$. In particular, we ensure that the left-to-right-ordered set of cliques has the property that vertices in a clique are only incident to vertices in the same clique and to the two neighboring cliques; see Figure~\ref{fig:unit_interval}(a). We then create an SOC diagram for each clique; see the red, blue and green layouts of the three  cliques in Figure~\ref{fig:unit_interval}(b).

		\begin{figure}[tp]
			\centering
			\includegraphics[width=\textwidth,page=3]{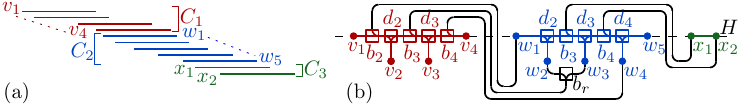}	
			\vspace{-0.3cm}
			\caption{(a)~A unit interval graph $G$ with a decomposition of its vertices into a set of cliques as described in the proof of Theorem~\ref{thm:unitIntervalInSC}; and (b)~a strict confluent layout of $G$ computed by the algorithm described in the proof of Theorem~\ref{thm:unitIntervalInSC}.\vspace{-0.5cm}	}
			\label{fig:unit_interval}
		\end{figure}		
		
		In order to realize the remaining edges we first make the following useful observation. Let $(v_1,\ldots,v_k)$ denote the vertices of some clique $C$ ordered from left to right according to $\Gamma_{UI}$. Then since all vertices are represented by unit intervals, if $v_i$ is incident to a vertex $w$ in the subsequent clique, also $v_j$ must be incident to $w$ for $i < j \leq k$. We use this observation to insert a split junction $b_i$ in the SOC diagram of $C$ such that all vertices with index at least $i$ can access a smooth arc that connects them with $w$; see the black arcs in Figure~\ref{fig:unit_interval}(b). We route arcs between cliques $C_i$ and $C_{i+1}$ first above clique $C_i$, then let it intersect with a line $H$ that passes through all the cliques (which intuitively inverts the ordering of such arcs) and then finish the drawing below clique $C_{i+1}$; refer to Figure~\ref{fig:unit_interval}(b) for an illustration. By adopting this scheme for each pair of consecutive cliques, intersections can be prevented.
	\end{proof}

	\section{Strict Bipartite-Outerconfluent Drawings}\label{sec:bipartite}
	Let $ G$ be a bipartite graph with bipartition $ (X, Y) $. 
	An outerconfluent drawing of $G$ is \emph{bipartite-outerconfluent} if the vertices in $X$ (and hence also $Y$) occur consecutively on the boundary. Graphs admitting such a drawing are called \emph{bipartite-outerconfluent}. The \emph{bipartite permutation} graphs are just the graphs that are bipartite and \emph{permutation} graphs, where a permutation graph is a graph that has an intersection model of straight lines between two parallel lines~\cite{pnueli1971transitive}.

	\begin{theorem}[Hui et al. \cite{hui2007train}]
		\label{thm:bp}
		The class of bipartite-permutation graphs is equal to the class of bipartite-outerconfluent graphs, i.e., the class of bipartite graphs admitting an intersection representation of straight-line segments between two parallel lines.
	\end{theorem}

	It is natural to consider the idea of bipartite drawings also in the strict outerconfluent setting. 
	We call a strict outerconfluent drawing $ D $ of $G$ \emph{bipartite} if 
	it is bipartite-outerconfluent and strict. 
	The graphs admitting such a drawing are called \emph{strict bipartite-outerconfluent graphs}. 
	In this section we extend Theorem~\ref{thm:bp} to the notion of strictness.
	The next lemma and observation are required in the proof of our theorem. The \emph{domino} graph is the graph resulting from gluing two 4-cycles together~at~an~edge.

	\begin{restatable}{lemma}{lemcsix}
		\label{lem:c6}
		Suppose that a reduced confluent diagram $ D = (N,J,\Gamma) $ contains two distinct uv-paths. Then we can find in $ G_D = (V_D,E_D) $ a set $ V' \subseteq V_D $ such that $ G[V'] $ is isomorphic to $ C_6 $ with at least one chord. 
	\end{restatable}

	\begin{restatable}{observation}{obsdom}
		\label{obs:dom}
		Let $ G = (V,E) $ be a graph and $ V' \subseteq V $ a subset of six vertices such that $ G[V'] $ is isomorphic to a domino graph and let $X\cup Y = V'$ be the corresponding bipartition. Now let $\pi$ be a cyclic order of $ V' $ in which the vertices in $ X $ and in $Y$ are contiguous, respectively. Then there is no strict outerconfluent diagram $ D = (N,J,\Gamma) $ with order $\pi$ and  $ G_D = G[V'] $ or, consequently, $G_D = G$.
	\end{restatable}
	
	\begin{restatable}{theorem}{thmbipartite}
		The (bipartite-permutation $ \cap $ domino-free)-graphs are exactly the strict bipartite-outerconfluent graphs.
	\end{restatable}
	
	\begin{proof}[Sketch]
		Let $ G = (V,E) $ be a (bipartite-permutation $ \cap $ domino-free) graph. By Theorem~\ref{thm:bp} we can find a bipartite-outerconfluent diagram $ D = (N,J,\Gamma) $ which has $ G_D = G $. Now assume that $ D $ is reduced but not strict. In this case we find six nodes $ N' \subseteq N $ corresponding to a vertex set $ V' \subseteq V_D $ in $ G_D$  such that $ G_D[V'] = (V', E') $ is a $ C_6 $ with at least one chord by Lemma~\ref{lem:c6}. In addition, since $ D $ (and hence also $G_D$) is bipartite and domino-free, we know there are two or three chords. Then $ G_D[V'] $ is  a $ K_{3,3} $ minus one edge $ e \in E' $ or $ K_{3,3} $. In a bipartite diagram these can always be drawn in a strict way.

		For the other direction, consider a strict bipartite-outerconfluent diagram $ D=(N,J,\Gamma) $. By Theorem~\ref{thm:bp}, $ G_D  $ is a bipartite permutation graph, and by Observation~\ref{obs:dom}, it must be domino-free. Thus, $ G_D  $ must be as described.
	\end{proof}

	\section{Strict Outerconfluent Graphs Have Cop Number Two}\label{sec:cops}
	The \emph{cops-and-robbers} game~\cite{af-gcr-84} on  a graph $ G = (V,E) $ is a two-player game with perfect information. The \emph{cop-player} controls $ k $ \emph{cop tokens}, while the \emph{robber-player} has one \emph{robber token}. In the first move the cop-player places the cop tokens on vertices of the graph, and then the robber places his token on another vertex. In the following the players alternate, in each turn moving their tokens to a neighboring vertex or keeping them at the current location. The cop-player is allowed to move all cops at once and multiple cops may be at the same vertex. The goal of the cop-player is to catch the robber, i.e., place one of its tokens on the same vertex as the robber. 
	
	The \emph{cop number} $\textit{cop}(G)$ of a graph $ G $ is the smallest integer $ k $ such that the cop-player has a winning strategy using $ k $ cop tokens. Gaven\u{c}iak et al. \cite{DBLP:conf/isaac/GavenciakJKK13} showed that the cop number of outer-string graphs is between three and four, while the cop-number of many other interesting classes of intersection graphs, such as circle graphs and interval filament graphs, is two. We show that the cop number of SOC graphs is two as well.
	
	Consider a SOC drawing $ D = (N,J,\Gamma) $ of a graph $G=(V,E)$, which we can assume to be connected. For nodes $ u,v \in N $, let the node interval $ N[u,v] \subset N $ be the set of nodes in clockwise order between $u$ and $v$ on the outer face, excluding $ u $ and $ v $. Let the cops be located on nodes $C\subseteq N$ and the robber be located on $r\in N$. We say that the robber is \emph{locked} to a set of nodes $N' \subset N $ if $r\in N'$ and every path from $r$ to $N\setminus N'$ contains at least one node that is either in $C$ or adjacent to a node in $C$; in other words, a robber is locked to $N'$ if it can be prevented from leaving $N'$ by a cop player who simply remains stationary unless the robber can be caught in a single move. The following lemma establishes that a single cop can lock the robber to one of two ``sides'' of a SOC drawing.
	
	\begin{restatable}{lemma}{lemblock}
		\label{lem:block}
		Let $ D = (N,J,\Gamma) $ be a SOC diagram of a graph $G$. Let a cop be placed on node $u$, the robber on node $r\neq u$ and not adjacent to $u$, and let $v\neq r$ be an arbitrary neighbor of $u$. Then the robber is either locked to $ N[u,v] $ or locked to $ N[v,u] $.
	\end{restatable}
	
	Let $ u,v \in N $ be two nodes of a SOC diagram $ D = (N,J,\Gamma) $. 
	We call a neighbor $ w $ of $u$ in $N[u,v] $ \emph{cw-extremal}  (resp.~\emph{ccw-extremal})
	for $ u,v $ (assuming such a neighbor exists), if it is the last neighbor of $ u $ in the clockwise (resp.~counterclockwise) 
	traversal of $ N[u,v] $.
	Now let $ u,v $ be two neighboring nodes in $N$, $ w \in N[u,v] $ be the cw-extremal node for $ u $ and $ x \in N[u,v] $ be the ccw-extremal node for $ v $. 
	If $ w $ appears before $ x $ in the clockwise traversal of $ N[u,v] $ we call $ w,x $ the \emph{extremal pair} of the $uv$-path, see Fig.~\ref{fig:cops_and_robbers}(b) and (c).
	In the case where only one node of $u,v$ has an extremal neighbor $w$, say $u$, we define the extremal pair as $v,w$.
	In the following we assume that for a given $uv$-path the extremal pair exists.
	
	\begin{figure}[tbp]
		\centering
		\includegraphics{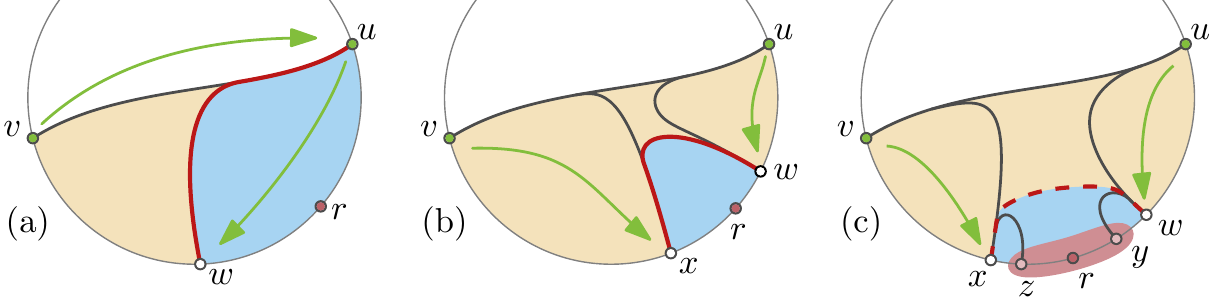}
		\caption{Moves of the cops to confine the robber to a strictly smaller range.}
		\label{fig:cops_and_robbers}
	\end{figure}

	\begin{restatable}{lemma}{lembicomponent}
		\label{lem:bi_component}
		Let $ D = (N,J,\Gamma) $ be a SOC diagram of a graph $ G $, $ u,v \in N $ be two nodes connected by a $ uv $-path in $ P(D) $ and $ w,x \in N[u,v] $ the extremal pair of the $ uv $-path. If the cops are placed at $ u $ and $ v $ and the robber is at $ r \in N[u,v] $, $r\neq w $, $ r \neq x $, there is a move that locks the robber to $ N[u,w] $, $ N[w,x] $ or $ N[x,v] $.
	\end{restatable}

	\begin{restatable}{lemma}{lembicatch}
		\label{lem:bi_catch}
		Let $ D = (N,J,\Gamma) $ be a SOC diagram of a graph $ G $, $ u,v \in N $ be two nodes connected by a $ uv $-path in $ P(D) $ and $ w,x \in N[u,v] $ be the extremal pair of the $ uv $-path such that there is no $wx$-path in $P(D)$. If the robber is at $ r \in N[w,x] $ and the cops are placed on $ w,x $ we can find $ y,z \in N[w,x] \cup\{w,x\} $ such that the $ yz $-path exists in $ P(D) $ and the robber is locked to $ N[y,z] $. %
	\end{restatable}

	Combining Lemmas~\ref{lem:block}, \ref{lem:bi_component} and~\ref{lem:bi_catch} yields the result.
	
	\begin{restatable}{theorem}{thmcopsandrobbers}
		\label{thm:cops_and_robbers}
		SOC graphs have cop number two. 
	\end{restatable}
	\begin{proof}[Sketch]
		Let $ D = (N,J,\Gamma) $ be a strict-outerconfluent diagram of a (connected) graph $G$. Choose any $uv$-path in $P(D) $ and place the cops on $ u $ and $ v $ as the initial move. The robber must be placed on a node $ r$ that is either in $N[u,v]$ or in $N[v,u]$; by symmetry, let us assume the former. By Lemma~\ref{lem:block}, the robber is now locked to $N[u,v] \neq \emptyset$.
		
		In every move we will shrink the locked interval until eventually the robber is caught. We distinguish three cases, based on the extremal neighbors $w$ and $x$ of $ u $ and $ v $ in $ N[u,v] $ and their ordering along the outer face. If $ w,x $ form no extremal pair, we can use Lemma~\ref{lem:block}, if they do form an extremal pair, we use first Lemma~\ref{lem:bi_component} and then, depending on the configuration, again Lemma~\ref{lem:block} (see Fig.~\ref{fig:cops_and_robbers}(b)) or go into the case of Lemma~\ref{lem:bi_catch} (see Fig.~\ref{fig:cops_and_robbers}(c)). 
	\end{proof}
	
	Theorem~\ref{thm:cops_and_robbers} suggests a closer link between SOC graphs and interval-filament graphs~\cite{gavril2000maximum},
	another subclass of outer-string graphs with cop number two.

	\section{Clique-width of Tree-like Strict Outerconfluent Graphs}\label{sec:cliquewidth}
	In 2005, Eppstein et al.~\cite{eppstein2005delta} showed that every strict confluent graph whose arcs in a strict confluent drawing topologically form a tree is distance hereditary and hence exhibits certain well-understood structural properties---in particular, every such graph has bounded \emph{clique-width}~\cite{CourcelleMakowskyRotics00}. 
	These graphs %
	are called \emph{$\Delta$-confluent} graphs. In their tree like confluent drawings an additional type of 3-way junction is allowed, the \emph{$\Delta$-junction}, which smoothly links together all three incident arcs. See Fig.~\ref{fig:delta_junctions}, where the junctions $ j' $ and $ k' $ now form a single $\Delta$-junction instead of three separate merge or split junctions.

	\begin{wrapfigure}[14]{r}{.4\textwidth}
		\centering
		\vspace{-1.25cm}
		\includegraphics[scale=1]{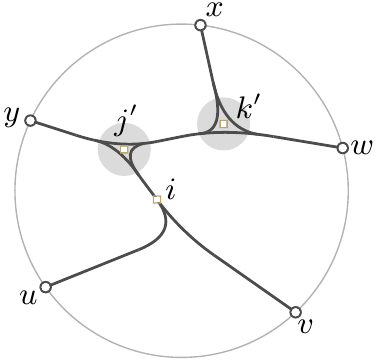}
		\caption{A $\Delta$-confluent diagram representing $K_5 - (u, v) $. Nodes are disks, junctions are squares. $ \Delta $-junctions are marked with a grey circle.}
		\label{fig:delta_junctions}
	\end{wrapfigure}
	
	In this section, we lift the result of Eppstein et  al.~\cite{eppstein2005delta} to the class of strict outerconfluent graphs: in particular, we show that as long as the arcs incident to junctions (including $\Delta$-junctions) topologically form a tree, strict outerconfluent graphs also have bounded clique-width. Equivalently, we show that ``extending'' any drawing covered by Eppstein et al.~\cite{eppstein2005delta} through the addition of outerplanar drawings of subgraphs in order to produce a strict outerconfluent drawing does not substantially increase the clique-width of the graph.
	Since the notion of clique-width will be central to this section, we formally introduce it below (see also the work of Courcelle et al.~\cite{CourcelleMakowskyRotics00}). %
	A $k$-graph is a graph whose vertices are labeled by $[k] = \{1, 2, \dots, k\}$; formally, the graph is equipped with a labeling function $\gamma\colon V(G) \to [k]$, and we also use $\gamma^{-1}(i)$ to denote the set of vertices labeled $i$ for $i \in [k]$. 
	We consider an arbitrary graph as a $k$-graph with all vertices labeled by $1$. 
	We call the
	$k$-graph consisting of exactly one vertex $v$ (say, labeled by $i$) an \emph{initial} $k$-graph and denote it by $i(v)$.
	The clique-width of a graph $G$ is the smallest integer $k$ such that $G$ can be constructed from \emph{initial} $k$-graphs by means of repeated application of the following three operations:
	\begin{enumerate}
		\item Disjoint union (denoted by $\oplus$);
		\item Relabeling: changing all labels $i$ to $j$ (denoted by $p_{i\rightarrow j}$);
		\item Edge insertion: adding an edge between every vertex labeled by $i$ and every vertex labeled by $j$, where $i\neq j$ (denoted by $\eta_{i,j}$ or $\eta_{j,i}$).
	\end{enumerate}
	The construction sequence of a $k$-graph $G$ using the above operations can be represented by an algebraic term composed of $i(v)$, $\oplus$, $p_{i\rightarrow j}$ and $\eta_{i,j}$ (where $v \in V(G)$, $i\neq j$ and $i,j\in [k]$). Such a term is called a \emph{$k$-expression} defining $G$, and the \emph{clique-width} of $G$ is the smallest integer $k$ such that $G$ can be defined by a $k$-expression. 
	Distance-hereditary graphs are known to have clique-width at most $3$~\cite{GolumbicR00} and outerplanar graphs have clique-width at most $5$ due to having treewidth at most~$2$~\cite{CourcelleO00,Baker94}. 

	Let \emph{(tree-like) $\Delta$-SOC graphs} be the class of all graphs which admit strict outerconfluent drawings (including $\Delta$-junctions) such that 
	the union of all arcs incident to at least one junction
	topologically forms a tree. 
	Clearly, the edge set $E$ of every tree-like $\Delta$-SOC graph $G=(V,E)$ with confluent diagram $D_G$ can be partitioned into sets $E_s$ and $E_c$, where $E_s$ (the set of \emph{simple edges}) contains all edges represented by single-arc paths in $D$ not passing through any junction and $E_c$ (the set of confluent edges) contains all remaining edges in $G$. Let $G_c=G[E_c]=(V_c,E_c)$ be the subgraph of $G$ induced by $E_c$, i.e., $V_c$ is obtained from $V$ by removing all vertices without incident edges in $E_c$. 
	
	We note that even though $G_c$ is known to be distance-hereditary~\cite{eppstein2005delta} and $G-E_c$ is easily seen to be outerplanar, this does not imply that tree-like $\Delta$-SOC graphs have bounded clique-width---indeed, the union of two graphs of bounded clique-width may have arbitrarily high clique-width (consider, e.g., the union of two sets of disjoint paths that create a square grid). 
	Furthermore, one cannot easily adapt the proof of Eppstein et al.~\cite{eppstein2005delta} to tree-like $\Delta$-SOC graphs, as that explicitly uses the structure of distance-hereditary graphs; note that there exist outerplanar graphs which are not distance-hereditary, and hence tree-like $\Delta$-SOC graphs are a strict superclass of distance hereditary graphs.
	Before proving the desired theorem, 
	we introduce an observation which will later allow us to construct parts of $G$ in a modular manner.
	
	\begin{restatable}{observation}{obscw}
		\label{obs:cw}
		Let $H=(V,E)$ be a graph of clique-width $k\geq 2$, let $V_1, V_2$ be two disjoint subsets of $V$, and let $s \in V \setminus (V_1 \cup V_2)$. Then there exists a $(3k+1)$-expression defining $H$ so that in the final labeling
		all vertices in $V_1$ receive label $1$, all vertices in $V_2$ receive label $2$, $s$ receives label $3$ and all remaining vertices receive label $4$.
	\end{restatable}

	\begin{restatable}{theorem}{thmcw}
		\label{thm:cw}
		Every tree-like $\Delta$-SOC graph has clique-width at most $16$.
	\end{restatable}

	\noindent \textit{Proof (Sketch).}
	We begin by partitioning the edge set of the considered $\Delta$-SOC graph into $E_c$ and $E_s$, as explained above, and by setting an arbitrary arc incident to a junction as the root $r$. Given a tree-like $\Delta$-SOC drawing of the graph, our aim will be to pass through the confluent arcs of the drawing in a leaves-to-root manner so that at each step we construct a $16$-expression for a certain circular segment of the outer face. This way, we will gradually build up the $16$-expression for $G$ from modular parts, and once we reach the root we will have a complete $16$-expression for $G$.

	At its core, the proof partitions nodes in the drawing into \emph{regions}, delimited by arcs connecting nodes and junctions (such nodes are \emph{not} part of any region). Each region is an outerplanar graph (which has clique-width at most 5), and furthermore the nodes in a region can only be adjacent to the nodes on the boundary of that region. Hence, by Observation~\ref{obs:cw} using $k=5$, each region can be constructed by a $16$-expression which also uses separate labels to capture the neighborhood of that region to its  border. See Fig.~\ref{fig:cliquewidth} for an illustration.

	\begin{wrapfigure}[14]{r}{.45\textwidth}
		\vspace{-.3cm}
		\centering
		\includegraphics[scale=1]{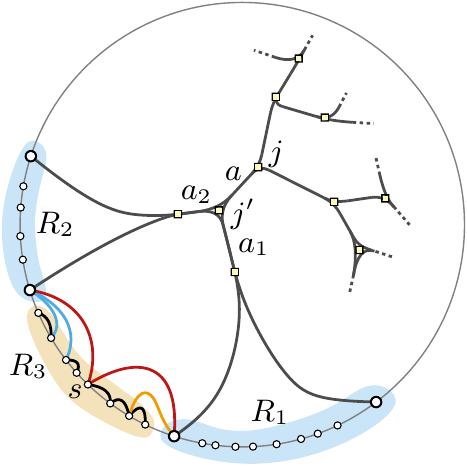}
		\caption{Sketch of a tree-like $\Delta$-SOC graph $G$ with its regions.}
		\label{fig:cliquewidth}
	\end{wrapfigure}
	The second ingredient used in the proof is tied to the tree-like structure of the drawing. In particular, one cannot construct a $16$-expression (and even any $k$-expression for constant $k$) by simply joining the regions together in the order they appear along the outer face. Instead, to handle the adjacencies imposed by the paths in the drawing, one needs to process regions (and their bordering vertices) in an order which respects the structure of the tree. To do so, we introduce a notion of \emph{depth}: nodes have a depth of $0$, while junctions have depth equal to the largest depth of its ``children'' plus $1$. Regions are then processed in an order which matches the depth of the corresponding junctions: for instance, if in Fig.~\ref{fig:cliquewidth} one of the junctions $a_1$ and $a_2$ has depth $d$ then junction $j'$ has depth $d+1$, and so the blue regions will be constructed by modular $16$-expressions before the yellow one. Afterwards, all three regions $R_1, R_2, R_3$ will be merged together into a blue region with a single $16$-expression. By iterating this process, upon reaching the root $r$ we obtain a $16$-expression that constructs the whole $\Delta$-SOC graph.
	\hfill $\qed$

	\section{Conclusion}
	While this work provides the first in-depth study of SC and SOC graphs, a number of interesting open questions remain. One such question is motivated by our results on the cop-number of SOC graphs: we showed that SOC graphs are incomparable to most classes identified to have cop number two by Gaven\u{c}iak et al.~\cite{DBLP:conf/isaac/GavenciakJKK13}, but we could not show such a result for the class of interval-filament graphs~\cite{gavril2000maximum}.
	It seems likely that SOC graphs are contained in this class. Similarly, it is open whether SC graphs are contained in subtree-filament graphs. Furthermore, it is conceivable that a similar construction  for the inclusion in string graphs, Section~\ref{sec:string}, could be used to show similar results for non-strict confluent graphs. Finally, investigating the curve complexity of our construction might provide insight into the curve complexity of SC and SOC diagrams.
	
	On the algorithmic side, Section~\ref{sec:cliquewidth} raises the question of whether clique-width might be used to recognize SOC graphs, and perhaps even for finding SOC drawings. Another decomposition-based approach would be to use so-called split-decompositions~\cite{DBLP:journals/dam/GioanP12}, which we did not consider here. It is also open whether all bipartite permutation and trapezoid graphs~\cite{golumbic1984tolerance,brandstadt1987bipartite} are SOC graphs. 
	Since bipartite permutation graphs are equivalent to bipartite trapezoid graphs~\cite{golumbic1984tolerance,brandstadt1987bipartite}, the former represents a promising first step in this direction. It also remains open if it is possible to drop the unit length condition on the intervals in Section~\ref{sec:interval}. We did not see an obvious way of adapting the construction for confluent drawings of interval graphs~\cite{degm-cdvndp-05}.

	\bibliographystyle{splncs04}
	\bibliography{paper}
	
		\clearpage
	
		\appendix
	\section{Omitted Proofs from Section~\ref{sec:string}}
		\obsindepsubtrees*
		\begin{proof}
			Assume $ i,j $ would be two such junctions, where $i$ is an ancestor of $j$ in~$T_u$. Then there are two distinct smooth paths from $v$ to $j$ in $D$: one passing through $i$ and then following the path to $j$ in $T_u$, the other one merging into $T_u$ in junction $j$. This contradicts the strictness of $D$. 
		\end{proof}
	
		\thmsoc*
		\begin{proof}
			Let $ G = (V,E) $ be a SOC graph with a strict outerconfluent drawing $ D = (N,J,\Gamma) $. Construct traces exactly as in the proof of Theorem~\ref{thm:sc} for the SC graphs. Since for each node $ u \in N$ the trace $ t(u) $ starts at the position of $ u $ in $ D $ we immediately know that every trace starts on the boundary of the enclosing disk of $ D $. Further, every trace is a continuous curve and does not leave the enclosing disk of $ D $ by construction. Hence, the constructed set of traces immediately yields an outer-string representation of $ G $. 
		\end{proof}
		
		\section{Omitted Proofs from Section~\ref{sec:interval}}
		\unitIntervalInSC*
	
	\begin{proof}
		Consider a unit interval graph $G$ with a unit-interval layout $\Gamma_{UI}$ of $G$. In $\Gamma_{UI}$, every vertex $v$ is represented by an interval $[\ell(v), r(v)]$ such that $r(v) - \ell(v) = u$ for some constant $u$, and two vertices $v,w$ are connected by an edge, if and only if $\ell(w) \in [\ell(v),r(v)]$ or $r(w) \in [\ell(v),r(v)]$. We can assume that all intervals have distinct endpoints, as otherwise the following modification is possible: Let $v_1, \ldots, v_k$ be $k$ vertices such that $\ell(v) := \ell(v_1) = \ldots = \ell(v_k)$ and let $\varepsilon > 0$. Then we assign new interval coordinates, such that $\ell(v_i) := \ell(v) + i \cdot \varepsilon$. Clearly, since all vertices $v_1, \ldots, v_k$ had the same neighborhood before, we can chose $\varepsilon$ sufficiently small to retain the incidencies. Let $O=(v_1,\ldots,v_n)$ denote the ordering of $V$ such that for vertex $v_i$ it holds that $\ell(v_i) < \ell(v_j)$ for all $j > i$. We call such an ordering $O$ a \emph{left-to-right-ordering}. 
		
		We first indentify subcliques $C_1,\ldots,C_k$ of $G$ such that each vertex is part of exactly one $C_i$ as follows. 
		We say that a vertex $v_i$ is a \emph{leader} in a left-to-right-ordering $O$ if and only if there exists no vertex $v_j$ with $j < i$ and $\ell(v_i) < r(v_j)$ such that $v_j$ is leader.
		Note that by definition $v_1$ is always a leader. 
		The second leader is the first vertex $v_i$ such that $r(v_1) < \ell(v_i)$ and so on. 
		Let $L=(l_1,\ldots,l_k) \subseteq V$ denote the left-to-right-ordered set of leaders. 
		It is easy to see that $L$ is a set of disjoint intervals. 
		We say that $v_j$ is \emph{leader of vertex $v_i$} or $v_j=\text{lead}(v_i)$ if $j \leq i$, $v_j$ is leader and there exists no leader $v_h$ for $j < h < i$. 
		Observe that every vertex $v$ has a uniquely defined leader which is the interval, in which $\ell(v)$ is located. 
		Since all intervals have unit size, there always exists an edge between vertices with the same leader. 
		Hence, all vertices with the same leader form a clique and we define $C_i = \{v \in V|\text{lead}(v)=l_i\}$. 
		Further, if $l_i$ is leader of vertex $v$, since all intervals are unit and leaders are disjoint, $r(v) \in (r(l_i),r(l_{i+1}))$, that is, vertex $v$ can only be connected to two leaders. For an illustration of such a decomposition refer to Figure~\ref{fig:unit_interval}(a).

		Next, we describe how to produce a strict confluent diagram $D$ of $G$. For an example illustration that follows the notation of the proof, refer to Figure~\ref{fig:unit_interval}(b).
		Let $\mathcal{C}=(C_1,\ldots,C_k)$ denote the left-to-right-ordered (according to their leaders) set of cliques. 
		We draw each clique $C_i=(V_i,V_i \times V_i) \in \mathcal{C}$ with the following SOC layout: Let $V_i = (v_1,\ldots,v_k)$ be the left-to-right-ordered set of vertices. 
		We position $v_1,d_2,\ldots,d_{k-1},v_k$ from left to right on a horizontal line $H$, and we position $v_i$ below $d_i$ for $2 \leq i \leq k-1$, where $d_i$ is a $ \Delta $-junction connecting $v_i$ with the two neighbors of $d_i$ on $H$. Note that a $\Delta$-junction smoothly links each pair of the three incident arcs.
		We order the drawings of all $C_i \in \mathcal{C}$ from left to right along $H$, that is the drawing of $C_i$ appears in between the drawings of $C_{i-1}$ and $C_{i+1}$. 
		Note that all vertices can be reached from below. 

		It remains to describe how to realize edges from vertices in $C_i \in \mathcal{C}$ to vertices in $C_{i+1}$.
		Let $C_i=(V_i,V_i \times V_i)$ and let $V_i = (v_1,\ldots,v_k)$ be the left-to-right-ordered set of vertices of $C_i$. 
		Consider edge $(v_\ell,w) \in E$ for $v_\ell \in V_i$ and $w \in V_{i+1}$. 
		Since$(v_\ell,w)$ exists, in $\Gamma_{UI}$, it holds that $\ell(v_\ell) < \ell(w) < r(v_\ell)$. 
		For all $v_j \in V_i$ with $j \geq \ell$, it obviously holds, that $\ell(v_j) \in  (\ell(v_\ell),\ell(w))$. 
		Therefore, also $(v_j,w) \in E$. 
		Further, for $v_\ell$ it clearly holds, that its neighbors $W \in V_{i+1}$ are  consecutive in the left-to-right-order of vertices defined by $\Gamma_{UI}$.  
		This allows us to realize the bundle of edges going from $V_i$ to a consecutive set of vertices $W \in V_{i+1}$ as follows. 
		Let $W \subseteq V_{i+1}$ such that $v_\ell$ is the leftmost neighbor of each $w \in W$. 
		We add a binary junction $b_\ell$ in between $d_\ell$ and $d_{\ell-1}$ (or $v_1$ if $\ell=2$) such that $b_\ell$ is a split junction for $d_\ell$ and a merge junction for $d_{\ell-1}$ and the root $b_r$ of a tree $T_b$ of binary junctions. 
		In $T_b$, each junction is a split junction for its ancestor and each leaf of $T_b$ is connected to a pair of vertices in $W$. 
		We position $T_b$ below $W$ and route the segment between $b_r$ and $b_\ell$ first above the drawing of $C_{i}$, then let it cross line $H$ and finally route it below the drawing of $C_{i+1}$. Since we use this scheme for all $C_i$, we avoid intersections of segments between different pairs of consecutive cliques. 
		Also, since we directly connect to vertices $W$ via $T_b$, we realize all edges exactly once, yielding a strict drawing of $G$.\end{proof}	
	
	\section{Omitted Proofs from Section~\ref{sec:bipartite}}\label{app:bipartite}	
		
		Let $ D = (N,J,\Gamma) $ be a strict outerconfluent diagram and $ \pi $ the clockwise cyclic order of the nodes. Consider $ G_D = (V_D,E_D) $ and use $ \pi $ to order the vertices of $ G_D $ on a circle. When using these vertex positions for a traditional straight-line circular layout, the order $\pi$ determines all the crossings between edges in $ G_D $. Namely two edges $ (a,c),(b,d) \in E_D $ intersect if and only if $ a < b < c < d $ in $ \pi $. We say that such a crossing is (confluently) \emph{representable} if the edges $ (a,b),(c,d) $ or $ (a,d),(b,c) $ are also contained in $ E_D $. If this is the case, then the induced subgraph of the four incident vertices of the crossing edges contains a $K_{2,2}$, which we can represent by a  merge-split pair. Otherwise, the crossing cannot be removed in an outerconfluent diagram with order $\pi$.

		\begin{restatable}[$\star$]{lemma}{lemrepresentable}
			\label{lem:representable}
			Let $ D $ be a strict outerconfluent diagram, $ G_D $ the graph represented by $ D $, and $ \pi $ the order inferred from $ D $. If we use $ \pi $ to create a circular straight-line layout of $G_D$, then every crossing is representable.
		\end{restatable}
		\begin{proof}
			Let $ (a,c),(b,d) \in E_D $ be two edges that have a crossing in a circular layout of $G_D$ with order $\pi$. Let $ p,q \in P(D) $ be two paths corresponding to these edges in $ D $. Since the order of the vertices and nodes is the same we know there exist two junctions $ i,j \in J $ such that either $a,b$ merge at $i$ and $c,d$ at $j$ or $a,d$ merge at $i$ and $c,b$ at $j$. This means that the edges $ (a,b) $ and $ (c,d) $ or $(a,d)$ and $(b,c)$ exist as well, so for every crossing we find that it is part of a $ K_{2,2} $. 
		\end{proof}
		
		It is clear that a graph can only have a strict outerconfluent drawing if it has a circular layout with all crossings representable. This is not sufficient  though, as there are such graphs that have no strict outerconfluent drawing. We obtain two 6-vertex obstructions for strict outerconfluent drawings, namely a $K_{3,3}$ with an alternating vertex order and a domino graph in bipartite order, see Fig.~\ref{fig:outer_circ_counter} and Fig.~\ref{fig:domino_counter}.

		\begin{restatable}[$\star$]{observation}{obskthreethree}
			\label{obs:k33}		
			Let $ G = (V,E) $ be a graph and $ V' \subseteq V $ a subset of six vertices such that $ G[V'] $ is isomorphic to $ K_{3,3} $ and let $X \cup Y = V'$ be the corresponding bipartition. Now let $\pi$ be a cyclic order of $ V' $ in which vertices from $ X $ and $ Y $ alternate, then there is no strict outerconfluent diagram $ D = (N,J,\Gamma) $ with order $ \pi $ and $ G_D = G[V'] $ or, consequently, $G_D = G$. %
		\end{restatable}
		\begin{proof}
			Let $ G = (V,E) $, $ V' $ and $ \pi $ as above. Draw $ G $ in a circular layout with the edges as straight lines. We then find the vertices in $ V' $ drawn as in Fig.~\ref{fig:outer_circ_counter}. Now, when creating the strict outerconfluent diagram $ D $ of $ G $, we at some point need to replace the crossings in $ G[V'] $ with confluent arcs and junctions. We can do this step-by-step.
			
			The first two crossings we can pick arbitrarily due to the symmetry of $K_{3,3}$ and replace them with two  merge-split pairs. In Fig.~\ref{fig:outer_circ_counter} we choose $ (x,z), (v,w) $ and $ (u,y),(x,z) $. This creates a diagram with one remaining crossing between two arcs. Since we replaced two representable crossings we did not change the adjacencies and hence this crossings can be replaced as well. The only possibility to replace the crossing with confluent junctions though is to put another merge-split pair. This, however, introduces a cycle into the diagram as shown in the second step in Fig.~\ref{fig:outer_circ_counter}. Now there are two distinct paths to go from $ u $ to $ v $, not counting multiple turns around the circular path, which contradicts strictness. 
		\end{proof}
	
	\obsdom*
		\begin{proof}
			Similarly to Observation~\ref{obs:k33} we need to transform a bipartite circular straight-line layout of the domino graph into an outerconfluent one by replacing crossings with merge-split junctions.
			For the order of vertices depicted in Fig.~\ref{fig:domino_counter}   it is clear that both crossings need to be replaced by merge-split pairs, but this introduces two distinct paths connecting $ u $ and $ v $, which contradicts strictness.
			One can observe that any other bipartite order of the vertices produces at least one non-representable crossing. Thus no outerconfluent drawing exists in such a case, regardless of strictness. 
		\end{proof}

		\begin{figure}[tb]
			\begin{minipage}[t]{.58\textwidth}
				\includegraphics[width=\textwidth]{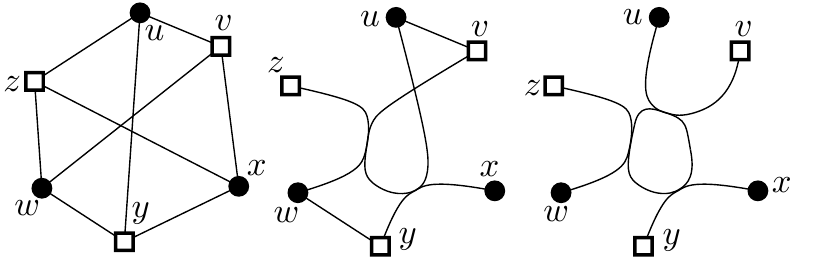}
				\caption{Forbidden alternating order of $ K_{3,3} $.}
				\label{fig:outer_circ_counter}
			\end{minipage}
			\hfill
			\begin{minipage}[t]{.4\textwidth}
				\includegraphics[width=\textwidth]{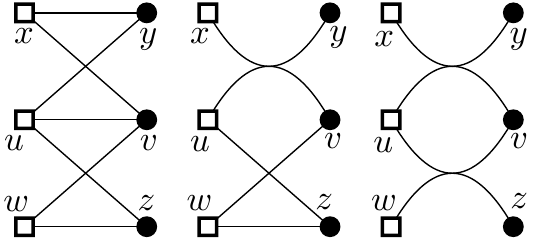}
				\caption{Forbidden domino order.}
				\label{fig:domino_counter}			
			\end{minipage}	
		\end{figure}

	\lemcsix*
		\begin{proof}
			Let $ p,q \in P(D) $ be two distinct $ uv $-paths in a reduced confluent diagram $ D = (N,J,\Gamma) $. We find two minimal distinct sub-paths $ p' $ of $ p $ and $ q' $ of $ q $ between two junctions $i,j$ of $p$ and $q$.
			
			First we assume that $i$ is a split-junction of $p,q$ and $j$ a merge-junction of $p,q$.		
			We claim that each of $p'$ and $q'$ must contain at least two junctions that form a merge-split pair. The argument is symmetric, so we focus on $p'$. If $p'$ passes through no junction, the arc of $p'$ can be removed since $q'$ achieves the same connectivity, but this contradicts that $D$ is reduced. If $p'$ passes through exactly one junction $i_1$, then the arc of $p'$ that does not split at $i_1$ can be removed without changing any node adjacencies in $G_D$. So $p'$ must contain at least two junctions $i_1$ and $i_2$. Assume there would be no merge-split pair. This means, coming from $i$ path $p'$ passes through a sequence of split-junctions followed by a sequence of merge-junctions. But in that case the arc connecting the last split-junction with the first merge-junction can be removed. So there must be a merge-split pair on $p'$ and similarly on~$q'$.
			
			Next we follow each arc of these two merge-split pairs that does not lie on $p'$ and $q'$ towards some reachable node. This yields four nodes $ x,y,w,z \in N$ that together with $u$ and $v$ form a domino subgraph as in Observation~\ref{obs:dom}, which is in fact a $C_6$ with a chord.
			
			Assume there are no two $uv$-paths as in the first case. Then $i$ is a merge-junction of $p,q$ and $j$ is a split-junction of $p,q$, see Fig.~\ref{fig:pointless_circle}. In this case, one of the paths, say $q'$ contains a cycle and visits $i$ and $j$ twice. Obviously all arcs of $p'$ are essential as removing them would disconnect $u$ and $v$. By the same arguments as before, there must be at least one merge-split pair on $q'$ as otherwise we can delete an arc of $q'$. However, if there is a single minimal merge-split pair $i_1, i_2$ on $ q' $ (i.e., two directly adjacent junctions), then one can reroute the arcs joining $q'$ in $i_1$ and $i_2$ towards $i$ and $j$, respectively and remove two arcs from $q'$, see Fig.~\ref{fig:domino_counter}. Hence there must be at least two merge-split pairs on $q'$ and we can find six nodes that form a $K_{3,3}$ subgraph as in Observation~\ref{obs:k33}, which is again a $C_6$ with at least one chord. 		
		\end{proof}	
		
				\begin{figure}[tbp]
				\centering
					\includegraphics{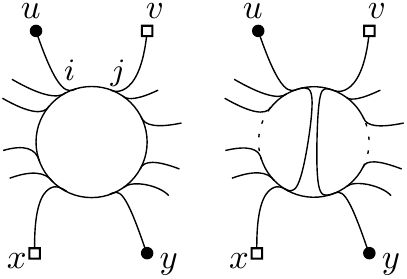}
					\caption{A circular path with only two merge-split pairs can be redrawn without change of the node order.}
					\label{fig:pointless_circle}
			\end{figure}
		
		\begin{figure}[tbp]
			\centering
			\includegraphics{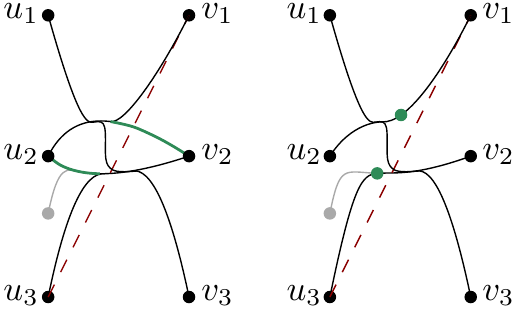}
			\caption{Redrawing a $ K_{3,3} $ minus an edge. The red edge is missing. Since the graph is bipartite, we find that for every path using one of the green arcs on the left, we can redraw it such that it merges into the path coming from $ u_3 $ at the green marker without creating any wrong adjacencies.}
			\label{fig:redrawbip}
		\end{figure}
		
		\thmbipartite*
		\begin{proof}
			Let $ G = (V,E) $ be a (bipartite-permutation $ \cap $ domino-free) graph. By Theorem~\ref{thm:bp} we can find a bipartite-outerconfluent diagram $ D = (N,J,\Gamma) $ which has $ G_D = G $. Now assume that $ D $ is reduced but not strict. In this case we find six nodes $ N' \subseteq N $ corresponding to a vertex set $ V' \subseteq V_D $ in $ G_D$  such that $ G_D[V'] = (V', E') $ is a $ C_6 $ with at least one chord by Lemma~\ref{lem:c6}. In addition, since $ D $ (and hence also $G_D$) is bipartite and domino-free, we know there are two or three chords. But then $ G_D[V'] $ is just a $ K_{3,3} $ minus one edge $ e \in E' $ or $ K_{3,3} $. In a bipartite diagram these can always be drawn in a strict way.
				
			Let $ V' = \{u_1,u_2,u_3,v_1,v_2,v_3\} $ where the $ u_i $ are on one side of the bipartition and the $ v_i $ on the other; they are ordered by their indices from top to bottom. First, observe that since $ G $ is a bipartite permutation graph and the algorithm by Hui et al.~\cite{hui2007train} uses the the strong ordering on the vertices, we get that in a $ K_{3,3} $ minus an edge this edge is either $ u_1v_3 $ or $ u_3v_1 $. Further we can assume that the non-strict doubled path is between $ u_2 $ and $ v_2 $ since $ D $ is reduced. For if this was not the case we would find that there are two merge-split pairs with vertices all below or above $ u_2,v_2 $ in $ D $, but then one of these pairs has junctions with both $ uv $-paths and we can simply reduce it.
				
			Now, w.l.o.g., assume $ u_3v_1 $ is the missing edge. It follows that $ u_1v_3 $ exists as an edge and the $ u_1v_3 $ path must also have junctions with both $ u_2v_2 $ paths. Furthermore, we know both these junctions are merge junctions for $ u_1 $ and $ u_2 $, $ v_2 $ and $ v_3 $ respectively. Thus we can redraw as in Fig.~\ref{fig:redrawbip}. For the case of no edge from $K_{3,3}$  missing the same argument applies.

			For the other direction, consider a strict bipartite-outerconfluent diagram $ D=(N,J,\Gamma) $. By Theorem~\ref{thm:bp}, $ G_D  $ is a bipartite permutation graph, and by Observation~\ref{obs:dom}, it must be domino free. Thus, $ G_D  $ must be as described.
		\end{proof}

		\section{Omitted Proofs from Section~\ref{sec:cops}}
		Consider a SOC drawing $ D = (N,J,\Gamma) $ of a graph $G=(V,E)$, which we can assume to be connected. For nodes $ u,v \in N $, let the node interval $ N[u,v] \subset N $ be the set of nodes in clockwise order between $u$ and $v$ on the outer face, excluding $ u $ and $ v $. Let the cops be located on nodes $C\subseteq N$ and the robber be located on $r\in N$. We say that the robber is \emph{locked} to a set of nodes $N' \subset N $ if $r\in N'$ and every path from $r$ to $N\setminus N'$ contains at least one node that is either in $C$ or adjacent to a node in $C$; in other words, a robber is locked to $N'$ if it can be prevented from leaving $N'$ by a cop player who simply remains stationary unless the robber can be caught in a single move. The following lemma establishes that a single cop can lock the robber to one of two ``sides'' of a SOC drawing.

		\lemblock*
		\begin{proof}
			Assume that, w.l.o.g., $r\in N[u,v]$, and consider an arbitrary path $P$ in~$G$ from $r$ to $N\setminus (N[u,v]\cup \{u,v\})$, which contains neither $u$ nor a neighbor of $u$. Consider the first edge $xy$ on $P$ such that $y\not \in N[u,v]$, and consider the $x$-$y$ path in $D$. Since $y\neq u$ and $y\neq v$, it must hold that the $x$-$y$ path in $D$ crosses the $u$-$v$ path at some junction. Hence $x$ must either be adjacent to $u$ or to $v$; in the former case, this immediately contradicts our assumption that $P$ contains no neighbor of $u$. In the latter case, it follows that there must be a junction on the $u$-$v$ path in $D$, which is used by the $x$-$y$ path to reach $y$, and hence $u$ must also be adjacent to $y$---once again contradicting our initial assumption about $P$.
		\end{proof}
	
		Let $ u,v \in N $ be two nodes of a SOC diagram $ D = (N,J,\Gamma) $. 
		We call a neighbor $ w $ of $u$ in $N[u,v] $ \emph{cw-extremal}  (resp.~\emph{ccw-extremal})
		for $ u,v $ (assuming such a neighbor exists), if it is the last neighbor of $ u $ in the clockwise (resp.~counterclockwise) 
		traversal of $ N[u,v] $.
		Now let $ u,v $ be two neighboring nodes in $N$, $ w \in N[u,v] $ be the cw-extremal node for $ u $ and $ x \in N[u,v] $ be the ccw-extremal node for $ v $. 
		If $ w $ appears before $ x $ in the clockwise traversal of $ N[u,v] $ we call $ w,x $ the \emph{extremal pair} of the $uv$-path, see Fig.~\ref{fig:cops_and_robbers}(b) and (c).
		In case only one node of $u,v$ has an extremal neighbor $w$, say $u$, we define the extremal pair as $v,w$.
		
		In the following we assume that for a given $uv$-path the extremal pair exists.
		
		\begin{figure}[tbp]
			\centering
			\includegraphics{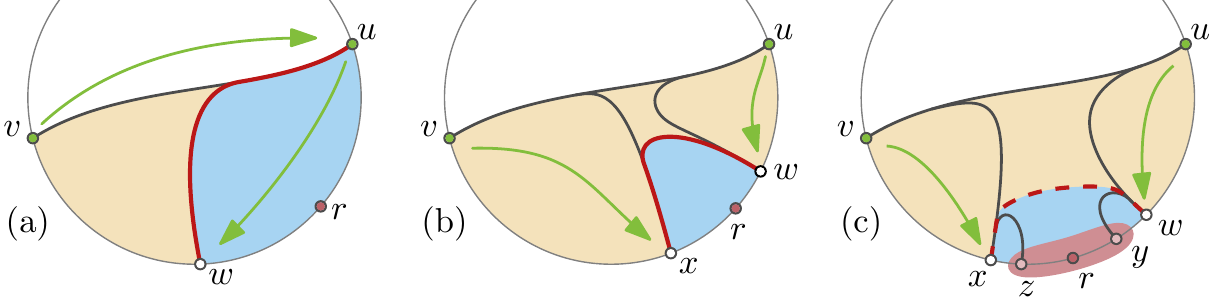}
			\caption{Moves of the cops to confine the robber to a strictly smaller range.}
			\label{fig:cops_and_robbers}
		\end{figure}

		\lembicomponent*
		\begin{proof}
			In case $ r \in N[u,w] $ or $ r \in N[x,v] $ we can swap the cops in one move (see Fig.~\ref{fig:cops_and_robbers}(a)) by moving the cop from $ v $ to $ u $ and from $ u $ to $ w $ in the former case and from $ u $ to $ v $ and $ v $ to $ x $ in the latter. This locks the robber to $ N[u,w] $ or $ N[x,v] $ by Lemma~\ref{lem:block}.
			
			The remaining case is $ r \in N[w,x] $. By construction of the extremal pair, no $y \in N[w,x]$ is a neighbor of $u$ or $v$. Because $G$ is a connected graph, there must be at least one $y \in N[w,x]$ that is a neighbor of $w$ or $x$ as the only smooth paths leaving $N[w,x]$ must share a merge junction with $u$ or $v$  on the paths towards $w$ or $x$. 
			In the next step, we move the cops from $u$ and $v$ to $w$ and $x$. 
			We need to distinguish two subcases. 
			If $r$ is not a neighbor of $w$ or $x$, then this position obviously locks the robber to $N[w,x]$ as any path leaving $N[w,x]$ must pass through a neighbor of $w$ or $x$.
			If, however, $r$ is already at a neighbor of $w\ne u$ or $x \ne v$, it may escape from $N[w,x]$ in the next move to a node in $N[u,w]$ or $N[x,v]$.
			But then by Lemma~\ref{lem:block} it locks itself to $N[u,w]$ or $N[x,v]$. Note that if  $w=u$ or $x=v$, then there is no way for $r$ to escape across $w$ or $x$, respectively, as $r$ would be a neighbor of $u$ or $v$ in that case.
		\end{proof}

		\lembicatch*
		\begin{proof}
			First assume that there is a path in $ G $ connecting $ w $ and $ x $, which passes only through nodes in $N[w,x]$; see Fig.~\ref{fig:cops_and_robbers}(c). 
			Let $ y \in N[w,x] $ be the ccw-extremal node of $ x $. If $ r \in N[y,x] $ we are done and by Lemma~\ref{lem:block} we can move the cop from $ w $ to $ y $ as the cop at $x$ suffices to lock the robber to $N[y,x]$.
			
			Now let  $ r \in N[w,y] $ instead and move the cop from $x$ to $y$. 
			As in the proof of the previous lemma, there are two subcases. 
			If $r$ is not a neighbor of $y$, the new position of the cops locks the robber to $N[w,y]$. 
			Otherwise, the robber might escape to $N[y,x]$ but immediately locks itself to $N[y,x]$ by Lemma~\ref{lem:block} and we are done.
			We repeat this process of going to the ccw-extremal node until we eventually lock the robber to some $N[y,z]$ where the $yz$-path is in $P(D)$.

			Now assume that there is no path from $ w $ to $ x $ in $ G $ that passes only through nodes in $N[w,x]$. 
			But then the only possibility for the robber to leave $N[w,x]$ is by passing through a neighbor of just one of $w$ or $x$, say $x$.
			We keep the cop at $x$, which suffices to lock the robber to $N[w,x]$.
			We can thus safely move the other cop first from $w$ to $x$ and from there following a path from one ccw-extremal node to the next until reaching a node $y$ such that the robber is now locked to $N[y,x]$ by the two cops.
			If there is an $xy$-path in $P(D)$ we are done.
			Otherwise we are now in the first case of the proof since by definition of $y$ there is a path from $x$ to $y$ in $G$ passing only through nodes in $N[y,x]$. 
		\end{proof}	
		
		Combining Lemma~\ref{lem:block}, \ref{lem:bi_component} and~\ref{lem:bi_catch} yields the result.
		\thmcopsandrobbers*
		\begin{proof}
			Let $ D = (N,J,\Gamma) $ be a strict-outerconfluent diagram of a (connected) graph $G$. Choose any $uv$-path in $P(D) $ and place the cops on $ u $ and $ v $ as initial turn. The robber must be placed on a node $ r$ that is either in $N[u,v]$ or in $N[v,u]$; by symmetry, let us assume the former. By Lemma~\ref{lem:block}, the robber is now locked to $N[u,v] \neq \emptyset$.
			
			In every move we will shrink the locked interval until eventually the robber is caught.
			Let $ w \in N[u,v]$ be the cw-extremal neighbor of $u$ and let $x \in N[u,v]$ be the ccw-extremal neighbor of $v$, which, for now, we assume to exist both. 
			If the clockwise order of $u,v,w,x$ on the outer face is $u<x<w<v$ then the robber must be either locked to $N[u,w]$ or to $N[x,v]$ by Lemma~\ref{lem:block}, so we move the cops to the respective nodes $u,w$ or $v,x$ and recurse with a smaller interval.
			If the clockwise order is $u<w<x<v$ then $w,x$ is an extremal pair and by Lemma~\ref{lem:bi_component} we can lock the robber to a smaller interval in the next move. 
			In case it is locked to $N[u,w]$ or to $N[x,v]$ we move the cops accordingly and recurse by Lemma~\ref{lem:block}.
			If it is locked to $N[w,x]$, then either there is a $wx$-path in $P(D)$ and again Lemma~\ref{lem:block} applies (Fig.~\ref{fig:cops_and_robbers}(b)) or there is no $wx$-path and Lemma~\ref{lem:bi_catch} applies after moving the cops to $w$ and $x$.
			
			It remains to consider the case that $u,v$ do not both have an extremal neighbor in $N[u,v]$. At least one of them, say $u$, must have a neighbor $w$ in $N[u,v]$---otherwise $G$ would not be connected.
			So we define the extremal pair as $v,w$.
			If the robber is in $N[u,w]$ we can move the cops to $u,w$ as in Fig.~\ref{fig:cops_and_robbers}(a) and apply Lemma~\ref{lem:block}. 
			If $r \in N[w,v]$, then we can move the cops to $v,w$ and Lemma~\ref{lem:bi_catch} applies.		
		\end{proof}

		\section{Non-Inclusion Results for Strict Outerconfluent Graphs}\label{sec:nonincl}
		\begin{figure}[tbp]
			\begin{minipage}[t]{.48\textwidth}
				\centering
				\includegraphics{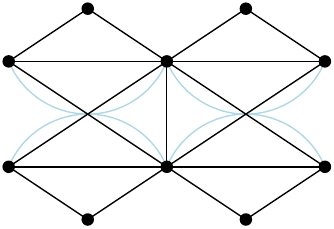}
				\caption{The black graph is not a circle graph, but it has a SOC drawing.}
				\label{fig:circle_counter}
			\end{minipage}
			\hfill
			\begin{minipage}[t]{.48\textwidth}
				\centering
				\includegraphics{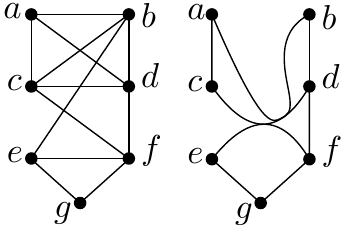}
				\caption{Counterexample for SOC $ \subseteq $ comparability.}
				\label{fig:comp_counter}
			\end{minipage}	
		\end{figure}
	
		\begin{figure}[tbp]
			\begin{minipage}[t]{.48\textwidth}
				\centering
				\includegraphics{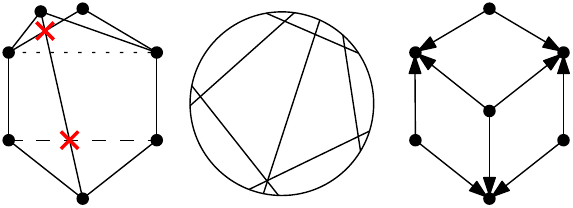}
				\caption{Without dashed and dotted edges, counterexample for comparability $ \subseteq  $ SOC, with dashed edge for circle $ \subseteq $ SOC, and with all edges for co-chordal $ \subseteq $ SOC.}			
				\label{fig:alternation_counter}
				\label{fig:circ_in_soc_counter}
			\end{minipage}		
			\hfill	
			\begin{minipage}[t]{.48\textwidth}
				\centering
				\includegraphics{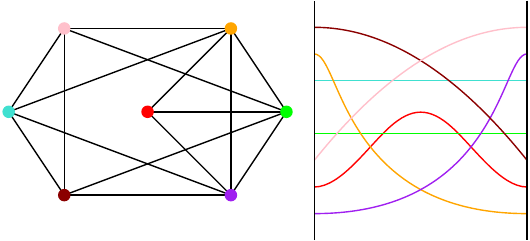}
				\caption{The complement of a subdivided star is a counterexample for co-comparability $ \subseteq $ SOC. Function representation on the right~\cite{GOLUMBIC198337}}
				\label{fig:co_comparability_in_soc_counter}
			\end{minipage}	
			\vspace{-0.5cm}
		\end{figure}
		
		Here we collect smaller results for
		classes of graphs which have non-empty intersection with the class of SOC graphs, but are neither superclasses nor subclasses. Theorem~\ref{thm:smaller} shows our incomparability results, while Corollary~\ref{cor:smaller} lists classes which are not contained in the class of SOC graphs. 
		Already Dickerson et al. \cite{degm-cdvndp-05} noticed that, given a non-planar graph $ G = (V,E) $, if for each edge $ (u,v) \in E $ we create a vertex $ w $ and make it adjacent to $u$ and $v$, then the resulting graph is not even confluent.
		
			\emph{Circle} graphs are graphs that are representable by an intersection model of chords in a circle. For the counterexample we use the following characterization of the class.
		
		\begin{definition}[Bouchet~\cite{bouchet1994circle}]
			\label{def:bouchet}
			The \emph{local complement} $ G*v $ is obtained from $ G $ by complementing the edges induced by $v$ and its neighborhood in $G$. Two graphs are said to be \emph{locally equivalent} if one can be obtained from the other by a series of local complements. %
			A graph $ G $ is a circle graph iff no graph locally equivalent to $ G $ has an induced subgraph isomorphic to $ W_5 $, $ BW_3 $, or $ W_7 $ (see Fig.~\ref{fig:bouchet_forbidden.}).
		\end{definition}
		
		\emph{Circular-arc} graphs are the graphs which have an intersection model of arcs of one circle \cite{hadwiger2015combinatorial}. Let $ ab $, $ cd $ be two non-crossing chords of a circle and $ a,b,c,d $ points on the circle in order $ a,b,c,d $. A \emph{circle trapezoid} then consists of the two chords $ ab $, $ cd $ and the circular-arcs $ bc $ and $ da $. A \emph{circle-trapezoid} graph is a graph which can be represented by intersecting circle trapezoids of one circle \cite{felsner1997trapezoid}. \emph{Co-comparability} graphs are the intersection graphs of $x$-monotone curves in a vertical strip \cite{GOLUMBIC198337}. The \emph{polygon-circle} graphs are the graphs which have an intersection model of polygons inscribed in the same circle \cite{kostochka1997covering}. \emph{Interval filament} graphs, defined by Gavril~\cite{gavril2000maximum}, are intersection graphs of continuous, non-negative functions defined on closed intervals, such that they are zero-valued at their endpoints.
		
		The \emph{comparability} graphs are the transitive orientable graphs. We also use the forbidden subgraph characterization by Gallai \cite{gallai1967transitiv}. The \emph{alternation} graphs are the graphs which have a \emph{semi-transitive orientation} \cite{halldorsson2011alternation}. For a graph $ G= (V,E) $ a semi-transitive orientation is acyclic and for any directed path $ v_1,\dots,v_k $ we either find $ (v_1,v_k) \not\in E $ or $ (v_i,v_j) \in E $ for all $ 1 \leq i < j < k $.
		
		The \emph{chordal} graphs are the graphs, which have no chord-less induced $ C_4 $ \cite{hajnal1958auflosung,golumbic2004algorithmic}. \emph{Co-chordal} graphs are the complement graphs of chordal graphs. \emph{Series-parallel} graphs are the graphs constructed from a multigraph, consisting of one vertex and a loop, by subdividing or replacing an edge repeatedly with two parallel ones \cite{DUFFIN1965303}.
		
		Finally the class of \emph{pseudo-split} graphs contains the graphs $ G = (V,E) $ such that $ V $ can be partitioned into three sets $ C $ (a complete graph), $ I $ (an independent set) and $ S $ (if not empty an induced $ C_5 $). Further every vertex in $ C $ is adjacent to every vertex in $ S $ and every vertex in $ I $ is non-adjacent to every vertex in $ S $.
		
		\begin{figure}[tbp]
			\centering
			\includegraphics{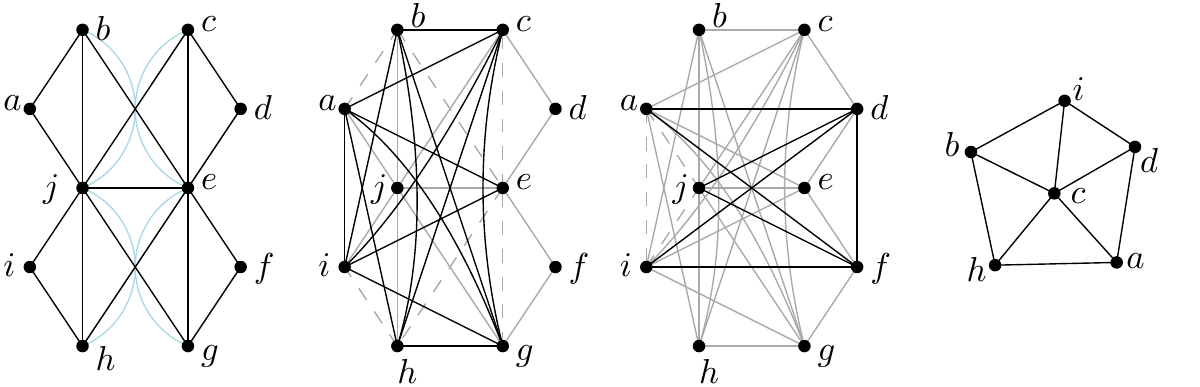}
			\caption{Counterexample for SOC $ \subseteq $ circle. The light-blue edges give the SOC drawing.}
			\label{fig:circ_counter}					
		\end{figure}	
		
		\begin{restatable}{theorem}{thmsmaller}
			\label{thm:smaller}
			These graph classes are incomparable to SOC graphs: circle, circular-arc, (co-)chordal, (co-)comparability, pseudo-split and series-parallel.
		\end{restatable}
		\begin{proof}

			\textbf{Circle.} Using the characterization of circle graphs due to Bouchet~\cite{bouchet1994circle}, we show that the graph in Fig.~\ref{fig:circle_counter} is not a circle graph. Further, the graph in Fig.~\ref{fig:circ_in_soc_counter} is a circle graph, but has no strict outerconfluent drawing.
			
			\textbf{Circular-arc.} Circular-arc graphs do not contain every complete bipartite graph, but obviously those have a strict-confluent drawing. Conversely this class contains $ W_5 $ and, as observed by Eppstein et al.~\cite{eppstein2016strict} $ W_5 $ is not a SOC graph.
			
			\textbf{Chordal.} $C_4$ is not a chordal graph by definition, but it is a SOC graph. Let $ G = (V,E) $ be a complete graph on five vertices. Attach for every $ (u,v) \in E $ a vertex $ w $ with edges $ (u,w) $ and $ (w,v) $. This graph is  chordal, but it does not even have a confluent drawing by \cite{degm-cdvndp-05}.
			
			\textbf{Co-chordal.} $ C_5 $ is self-complementary and has a strict outerconfluent drawing, but co-chordal graphs do not contain the complement of $C_{n+4} $. Adding the dashed edge in Fig.~\ref{fig:circ_in_soc_counter} makes the graph co-chordal, but the crossings are not representable, so there is no strict outerconfluent drawing.
			
			\textbf{Comparability.} The graph in Fig.~\ref{fig:comp_counter} has a strict outerconfluent drawing, but is among the forbidden subgraphs of the class of comparability graphs \cite{gallai1967transitiv}. Any order of $ BW_3 $ has a crossing that is not representable. So it has no strict outerconfluent drawing, but $ BW_3 $ has a transitive orientation as shown in Fig.~\ref{fig:alternation_counter}. 	
			
			\textbf{Co-comparability.} $ C_5 $ is not a co-comparability graph, but it is a SOC graph. Fig.~\ref{fig:co_comparability_in_soc_counter} shows a graph which is verified to not be SOC by exhaustively searching all orders for represented crossings, but it has an intersection representation of $ x $-monotone curves between two parallel lines~\cite{GOLUMBIC198337}.
			
			\textbf{Pseudo-split.} $C_4$ is not a pseudo split graph, but it is a SOC graph. Any order of $ W_5 $ has a crossing that is not representable. So it has no outerconfluent drawing, but $ W_5 $ is a pseudo-split graph by definition (take the central vertex as the clique and the other five vertices as the $ C_5 $).
			
			\textbf{Series-parallel.} $ K_4 $ is not a series-parallel graph by definition, but it is a SOC graph. Let $ G = (V,E) $ be a domino graph. Subdivide the chord $ (u,v) $ with a vertex $ w $. This is a series-parallel graph, but any crossing between $ (u,w) $ or $ (v,w) $ and another edge cannot be represented.
		\end{proof}
		
		\begin{restatable}[$ \star $]{corollary}{corsmaller}
			\label{cor:smaller}
			These graph classes are not contained in the class of SOC graphs: alternation, circle-trapezoid, polygon-circle, interval-filament, subtree-filament, \mbox{(outer-)string}.
		\end{restatable}
		\begin{proof}
		Comparability graphs are known to be contained in alternation graphs~\cite{halldorsson2011alternation}, but comparability graphs are not a subclass of SOC graphs. All other classes follow directly, since they are known to be superclasses of circle graphs.
	\end{proof}
	
	\begin{figure}[tbp]
		\centering
		\includegraphics{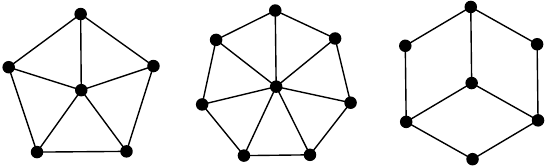}
		\caption{From left to right: the graphs $W_5$, $W_7$ and $BW_3$ used in Def.~\ref{def:bouchet}.}
		\label{fig:bouchet_forbidden.}
	\end{figure}
	
	\section{Omitted Proofs from Section~\ref{sec:cliquewidth}}
	\obscw*
	\begin{proof}
	Consider an arbitrary $k$-expression of $H$ which ends by setting all labels in $H$ to $1$. Now adjust the $k$-expression as follows: whenever a vertex in $V_1$ receives a label $i$, replace it with $i+k$, and whenever a vertex in $V_2$ receives a label $i$, replace it with $i+2k$, and use a special label $3k+1$ for $s$. This new $(3k+1)$-expression constructs $H$ and assigns all vertices in $V\setminus (V_1\cup V_2\cup \{s\})$, $V_1$, $V_2$ and $\{s\}$ the labels $1$, $k+1$, $2k+1$, and $3k+1$, respectively. To complete our construction, we merely permute the labels as required.
	\end{proof}

	\thmcw
	\noindent \textit{Proof.}
	Let us consider an arbitrary tree-like $\Delta$-SOC graph $G=(V,E)$ and let us fix a tree-like strict outerconfluent drawing $ D = (N,J,\Gamma) $ of $G$; let $\Gamma_c$ be the set of arcs with at least one endpoint in $J$. Based on $D$, we partition $E$ into the edge sets $E_c$ and $E_s$ as above. Let $V_c$ be the set of vertices incident to at least one edge in $E_c$ and let $G_c=(V_c,E_c)$.
	
	Note that $D_c=(N,J,\Gamma_c)$ is topologically equivalent to a tree (plus some singletons in $V \setminus V_c$), and let us choose an arbitrary arc in $\Gamma_c$ as the \emph{root} $r$. Our aim will be to pass through $D_c$ in a leaves-to-root manner so that at each step we construct a $16$-expression for a certain circular segment of the outer face. This way, we will gradually build up the $16$-expression for $G$ from modular parts, and once we reach the root we will have a complete $16$-expression for $G$.
	
	Our proof will perform induction along a notion of \emph{depth}, which is tied to the tree-like structure of $D_c$. We say that each node corresponding to a vertex in $V_c$ has \emph{depth} $0$, and we define the depth of each junction $j$ as follows: $j$ has depth $\ell$ if $\ell$ is the minimum integer such that at least two of the arcs incident to $j$ lead to junctions or nodes of depth  at most $\ell-1$ (for example, a junction of depth $1$ has $2$ arcs leading to nodes, while a junction of depth $2$ has at least one arc leading to a junction of depth $1$ and the other arc leads to either a node or another junction of depth $1$). We call an arc between a junction $j$ of depth $i$ and a junction (or node) of depth smaller than $i$ a \emph{down-arc} for $j$.

	Another notion we will use is that of a \emph{region}, see Fig.~\ref{fig:cliquewidth_appendix}: the region defined by a junction $j$ and one of its down-arcs $a$ is the segment of the boundary of the outer face delimited by the ``right-most'' and ``'left-most'' paths (not necessarily smooth), which leave $j$ through $a$.
	Crucially, we observe that the set $V_R$ of all vertices corresponding to nodes in a region $R$ can be partitioned into the following four groups:
	\begin{enumerate}
	\item[\textbf{A.}] one vertex on the left border of $R$;
	\item[\textbf{B.}] up to one vertex that is not on the left border but on the right border of $R$;
	\item[\textbf{C.}] vertices not on the border which have no neighbors outside of $R$;
	\item[\textbf{D.}] vertices not on the border which have at least one neighbor outside of $R$;
	\end{enumerate}
	and furthermore we observe that all vertices of group \textbf{D} have precisely the same neighborhood outside of $R$ (in particular, they must all have a path to $j$ which forms a smooth curve). In the degenerate case of nodes (which have depth $0$), we say that the region is merely the point of that node (and the corresponding vertex then belongs to group \textbf{A}).
	
	\begin{figure}[tbp]
		\centering
			\includegraphics[scale=1]{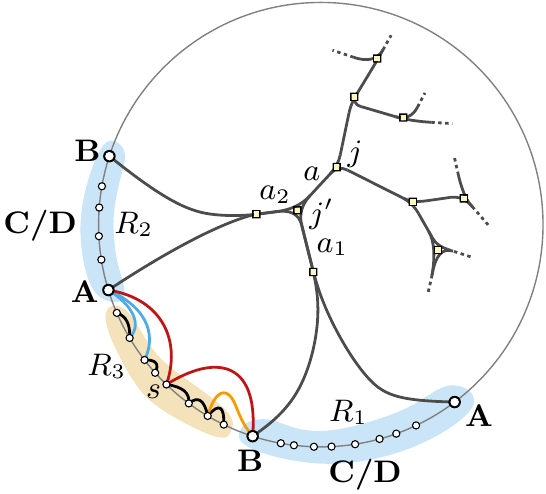}
		\caption{Sketch of a tree-like $\Delta$-SOC graph $G$ with the regions and junctions used in the inductive construction of the 16-expression defining $G$.}
		\label{fig:cliquewidth_appendix}
	\end{figure}

	As the first step of our procedure, for each $v\in V_c$ we create a $1$-expression $1(v)$ (i.e., we create each vertex in $V_c$ as a singleton). For the second step, we apply induction along the depth of junctions as follows. As our inductive hypothesis at step $i$, we assume that for each junction $j'$ of depth at most $i-1$ and each of its down-arcs defining a region $R'$, there exists a $16$-expression which constructs $G[V_{R'}]$ and labels $V_{R'}$ by using labels $1,2,3,4$ for vertices in groups $\textbf{A}$, $\textbf{B}$, $\textbf{C}$, $\textbf{D}$, respectively. We observe that the inductive hypothesis holds at step $1$: indeed, all regions at depth $0$ consist of a node, and we already created the respective $1$-expressions for all such nodes. 
	
	Our aim is now to use the inductive hypothesis for $i$ to show that the inductive hypothesis also holds for $i+1$---in other words, we need to obtain a $16$-expression which constructs and correctly labels the graph $G[V_{R}]$ for the region $R$ defined by each junction $j$ of depth $i$ and down-arc $a$. Assume that $a$ is incident to a junction $j'$ with down-arcs $a_1$ and $a_2$, defining the regions $R_1$ and $R_2$, respectively. By our inductive assumption, $G[V_{R_1}]$ and $G[V_{R_2}]$ both admit a $16$-expression which labels the vertices based on their group in the desired way. Now observe that $R$ is composed of the following parts: region $R_1$ on the ``left'', region $R_2$ on the ``right'', and a segment $R_3$ on the boundary of the outer face between $R_1$ and $R_2$. Crucially, we make the following observations for vertices $V_{R_3}$ in $R_3$:
	\begin{itemize}
	\item none of the vertices in $V_{R_3}$ are incident to an edge in $E_c$;
	\item $G[V_{R_3}]$ is outerplanar;
	\item at most one vertex, denoted $s$, in $V_{R_3}$ has two neighbors outside of $V_{R_3}$, notably the rightmost vertex in $V_{R_1}$ and the leftmost vertex in $V_{R_2}$;
	\item all vertices other than $s$ in $V_{R_3}$ either have no neighbors outside of $V_{R_3}$, or have one neighbor outside of $V_{R_3}$---in particular, either the rightmost vertex in $V_{R_1}$ or the leftmost vertex in $V_{R_2}$.
	\end{itemize}
	
	At this point, we can finally invoke Observation~\ref{obs:cw}. In particular, since $G[V_{R_3}]$ is outerplanar, it has clique-width at most $5$, and by using the observation we can construct a $16$-expression which labels all vertices adjacent to the right border of $R_1$ with label $1$, all vertices adjacent to the left border of $R_2$ with label $2$, vertex $s$ with label $3$, and all other vertices with label $4$. Now all that remains is to: 
	\begin{enumerate}
	\item relabel labels $1$--$4$ used in the $16$-expression for $R_2$ to labels $5$--$8$ and the labels $1$--$4$ used in the $16$-expression for $R_3$ to labels $9$--$12$, respectively;
	\item use the $\oplus$ operator to merge these $16$-expressions, 
	\item use the $\eta_{i,j}$ operator to add edges between $V_{R_1} \cup V_{R_2}$ and $V_{R_3}$ as required, in particular: $\eta_{9,2}$, $\eta_{11,2}$, $\eta_{10,5}$, $\eta_{11,5}$;
	\item use the $\eta_{4,8}$ operator to add all pairwise edges between the groups \textbf{D} of $V_{R_1}$ and $V_{R_2}$ in case junction $j'$ smoothly connects arcs $a_1$ and $a_2$;
	\item use the $p_{i\rightarrow j}$ operator to relabel as required by the inductive assumption, where depending on the junction type of $j'$ group \textbf{D} of $V_R$ either consist of the union of the groups \textbf{D} of $V_{R_1}$ and $V_{R_2}$ or it is identical to group \textbf{D} of just one of them. Group \textbf{A} coincides with group \textbf{A} of $V_{R_1}$ and group \textbf{B} coincides with group \textbf{B} of $V_{R_2}$. The remaining vertices form group \textbf{C}.
	\end{enumerate}
	
	The inductive procedure described above runs until it reaches the root arc $r$, and it is easy to observe that at this point we have constructed two $16$-expressions corresponding to the two regions, say $R_1^*$ and $R_2^*$, defined by paths which start at $r$ and go in the two possible directions. The two remaining regions on the outer face between $R_1^*$ and $R_2^*$ are then handled completely analogously as the regions denoted $R_3$ in our inductive step. Hence we conclude that there indeed exists a $16$-expression which constructs $G$.
	\hfill \qed

\end{document}